\documentclass{ws-ijgmmp-web}
\usepackage{bm}
\usepackage{color}
\usepackage[colorlinks=true]{hyperref}

\newcommand{\Ignore}[1]{}

\newcommand{\up}{{\uparrow}}
\newcommand{\down}{{\downarrow}}

\renewcommand{\Re}{\operatorname{Re}}
\newcommand\tr{\operatorname{tr}}

\def\e{\mathrm{e}}

\def\d{\mathrm{d}}

\def\H{\mathcal{H}}

\begin{document}

\markboth{P. Facchi, G. Garnero}
{Quantum Thermodynamics and Canonical Typicality}
%
\catchline{}{}{}{}{}
%

\title{QUANTUM THERMODYNAMICS\\ AND CANONICAL TYPICALITY}

\author{PAOLO FACCHI and GIANCARLO GARNERO}
\address{Dipartimento di Fisica and MECENAS,
Universit\`a di Bari, I-70126 Bari, Italy\\ 
INFN, Sezione di Bari,  I-70126 Bari, Italy}


\maketitle

\begin{abstract}
We present here a set of lecture notes on  quantum thermodynamics and canonical typicality. Entanglement can be constructively used in the foundations of statistical mechanics. An alternative version of the postulate of equal a priori probability is derived making use of some techniques of convex geometry.

\end{abstract}


\setcounter{section}{-1}
\section{Introduction} 

\label{sec:intro}

\textit{We have the general question of finding out those features which are common to almost all possible states of the assembly so that one may safely contend that they ``almost always'' obtain.}
\vspace{-0.5cm}
{\flushright E. Schr\"odinger~\cite{schr1}\\
}

\vspace{0.3cm}

We present here the notes of three lectures given by one of us at the workshop ``Quantum Physics: Foundations and Applications", held in February 2016 in Bangalore at the Centre for High Energy Physics, Indian Institute of Science.

The foundations of statistical mechanics are still a subject of debate. One of the most controversial issue is the validity of the postulate of equal a priori probability, which cannot be proved. In these notes we are going to discuss a set of ideas based on typicality
put forward by several authors~\cite{gemmer,goldstein,popescu,popescun},
who have been looking for a different approach.

The proposal will be to abandon the unprovable aforementioned postulate and to replace it with canonical typicality, which can be proved by means of the entanglement between a physical system and its environment.

These notes will be organized as follows. In the first lecture  we are going to introduce the postulate of equal a priori probability and we are going to discuss its inception in the foundations of statistical mechanics. The idea of an ensemble as a collection of identical systems will be introduced and the postulate of equal a priori probability will be translated into the choice of a particular ensemble, the microcanonical one. A quantum version of the aforesaid postulate will be presented in terms of the random phase postulate, and the derivation of the canonical ensemble for a weakly interacting system is presented.
 
In the second lecture entanglement will come on stage and we are going to present the phenomenon of canonical typicality, where thermalization emerges as a consequence of typicality. We learnt this approach from the works of S.~Popescu, A.~J.~Short and A.~Winter~\cite{popescu,popescun},
and we follow them fairly closely. We will be giving quantitative arguments in order to convince the reader about the validity of this result and we will try to explain all the needed mathematical tools.

This set of notes will end with a brief introduction to convex geometry in high dimensions and its counterintuitive results such as the phenomenon of concentration of measure. The lemma by L\'evy, which was used in the previous lecture in the proof of canonical typicality, will be framed and proved in this context.  

\section{Lecture 1: The Postulate of Equal a Priori Probability}
\subsection{Introduction}
\label{ergodic hypothesis}

Statistical mechanics has proved to be a very fruitful theory when dealing with systems with a large number of degrees of freedom.
Nevertheless its foundations are still a matter of strong debates. In this first lecture we would like to provide the reader with a short account on the  postulate of equal a priori probability, and its primary role in the formulation of statistical mechanics.
We will follow the review article~\cite{uffink}. 

We recall that \emph{Statistical Physics is a branch of physics intended to describe the thermal behavior and properties of macroscopic bodies, i.e. formed of a very large number of individual constituents, in relation to the microscopic dynamics of those constituents.}

As a matter of fact statistical physics historically stemmed from the thrust of phenomenological thermodynamics. In fact, by the mid of the 19th century, thermodynamical observables, such as temperature, heat and entropy, where linked to each other by means of empirical principles. 
The behavior at  thermal  equilibrium of macroscopic bodies was not deduced from general assumptions on the microscopic constituents of matter, rather it was framed into a set of effective laws, known as the \textit{laws of thermodynamics}. In particular, the \textit{second law of thermodynamics} restrained the class of natural phenomena which could actually happen.

It was only under the influence of 
Boltzmann, Gibbs and Maxwell that statistical physics emerged as a fundamental theory in the description of the physical world.
Statistical physics differed from phenomenological thermodynamics for two reasons:
first, a mechanical hypothesis on the microscopic behavior of the constituents of matter; and second, the introduction of probability and statistics in order to deduce the laws of thermodynamics.

Let us start from the assumption on the microscopic world. By the end of the 19th century the atomic theory (which had been strongly supported by Boltzmann and largely unaccepted by the rest of the scientific community) was becoming more and more popular. For this reason it started to be rather natural to consider, for example, an ideal gas confined in a box as a system of bouncing balls scattering with each other and with the box's walls.
Classical mechanics, which had scientifically ruled the physics of the 18th century, was once more put on trial, and
its power and effectiveness in the description of physical phenomena was questioned.

Consider a mechanical system of $N$ particles subject to a time-independent potential $V$. Suppose that every particle has mass $m$ and indicate by $(q,p)$ the $2N$ canonical coordinates on phase space $\Gamma$, say $(q,p)=(q_1,...,q_N,p_1,...,p_N)$. The Hamiltonian function  
describing the system is
\begin{equation}
\label{hamilt}
H(q,p)= \frac{1}{2m}\sum_{i=1}^N p_i^2 + V(q_1,...,q_N).
\end{equation}
The state of the system is described by a point in the phase space $\Gamma$, and its time evolution follows the dynamics induced by the vector field associated to the Hamiltonian function. Moreover, since energy $E$ is conserved, the evolution is confined to a level set of the Hamiltonian, say 
\begin{equation}
\Sigma_E=\{(q,p)\in \Gamma\,|\,H(q,p)=E\,\}.
\label{eq:energy shell}
\end{equation}

Here comes on the stage a second ingredient: the introduction of probability in terms of mechanical considerations. 
In~\cite{Boltz} Boltzmann proposed to interpret the probability associated to a particular state as the relative time (compared to a long period of measurement) spent by the system in that state. More concretely, the probability that the phase point lies in an infinitesimal region of $\Sigma_E$ is
\begin{equation}
\rho(q,p)\,\mathrm{d}q \mathrm{d}p= f(q,p)\,\delta\big(H(q,p)-E\big)\,\mathrm{d}q \mathrm{d}p,
\end{equation}
where $\mathrm{d}q\mathrm{d}p$ is the Lebesgue measure on $\Gamma$ and $f$ is a suitable function.
Since Boltzmann's intention was to describe an equilibrium situation, it seems likely that he assumed the former distribution to be time-independent.
This assumption relies on the hypothesis that the total time of a measurement is extremely long, almost infinite, if compared with the intrinsic time scales of evolution.

By Liouville's theorem, $f$ is a constant function on all the admissible trajectories on $\Sigma_E$. 
Moreover, if we assume (\emph{ergodic hypothesis}) that the trajectory of a single point in phase space fills densely $\Sigma_E$, then $f$ is a constant function on the whole $\Sigma_E$, say
\begin{equation}
\rho(q,p)\propto \delta\big(H(q,p)-E\big).
\end{equation}
By means of this argument Boltzmann was, then, able to prove that the thermal equilibrium can be described in terms of Maxwell's distribution of velocities.
Even so, Boltzmann was highly unsatisfied with the ergodic hypothesis and slowly abandoned it. It seems that he only considered it as a useful assumption for his general result: as long as the ergodic hypothesis holds the equilibrium state of an ideal gas is described by the Maxwell-Boltzmann distribution.

\subsection{The postulate of Equal a Priori Probability and the Microcanonical Ensemble}

A different path was laid by Gibbs who, making use of a different idea of probability, introduced the so-called ``ensembles".
In his book~\cite{gibbs} Gibbs did not introduce probability as an ingredient associated to the state of a single system,  but rather as a distribution function on a collection of identical systems, that is, an \emph{ensemble}.

One considers all the possible microscopic configurations of the system, i.e.\ points in phase space  (also called \emph{microstates}), which are compatible with a single macroscopic configuration (\emph{macrostate}). In this  approach, one is not interested in following the temporal evolution of a single microscopic configuration, but rather is concerned about the distribution of all the available microscopic configurations.

More concretely an ensemble is introduced as a probability density function on the phase space $\Gamma$, such that the average number of microstates in a region $A$ of $\Gamma$ is nothing but: $\int_A\rho(q,p)\d q \d p$. Moreover, the expectation value of an observable $f:\Gamma\to\mathbb{R}$ is the average of $f$ over $\Gamma$, that is:
\begin{equation}
\label{ensavg}
\langle f\rangle=\int_{\Gamma}f(q,p)\rho(q,p)\,\d q \d p.
\end{equation}
In this approach time enters into the evolution of $\rho$ which is dictated by  Liouville's theorem, that is:
\begin{equation}
\frac{\partial\rho_t}{\partial t}=\{H, \rho_t\},
\end{equation}
where  $\{f,g\}=\sum_i (\partial_{q_i} f \partial_{p_i} g - \partial_{p_i} f \partial_{q_i} g)$ are the Poisson brackets between the observables $f$ and $g$. The condition of \textit{statistical equilibrium} is provided by stationary ensembles:
\begin{equation}
\frac{\partial\rho_t}{\partial t}=0,
\end{equation}
and among these we find the microcanonical ensemble, which is the only one compatible with the conservation of energy, say:
\begin{equation}
\rho_{\mathrm{mc}}(q,p)= \frac{1}{|\Sigma_E|} \delta\big(H(q,p)-E\big)
\label{eq:microens}
\end{equation}
where $|\Sigma_E|$ is the  measure of $\Sigma_E$.

Classical statistical mechanics is founded on the \emph{Postulate of Equal a Priori Probability}.
It states that  the microstates accessible to an isolated system are all equally probable, because there is no evidence that certain microstates should be more probable than others. In other words when a macroscopic system is at  equilibrium, every state compatible with the energy restriction is equally available. Mathematically this translates into the choice of a constant density function, called the \emph{microcanonical ensemble}.

If the total energy is fixed up to an uncertainty $\delta$, assumed small (on a macroscopic scale), instead of~\eqref{eq:microens} one can consider a density
\begin{equation}
\rho_{\mathrm{mc}}(q,p)= \frac{1}{|\Gamma_{E}|} \chi_{[E, E+ \delta]} \big(H(q,p)\big),
\label{eq:microens2}
\end{equation}
where $|\Gamma_{E}|$ is the volume of the energy shell $\Gamma_{E}=\{(q,p) \, : \, E\leq H(q,p)\leq E+\delta\}$, and $\chi_\Omega$ is  the characteristic function of the set $\Omega$ [$\chi_\Omega(x) =1$ if $x\in\Omega$ and $=0$ otherwise]. 

\subsection{The Ergodic Hypothesis}
In what follows we are going to review with a few more details the ergodic hypothesis used by Boltzmann as a former argument to get the postulate of equal a priori probability. In a nutshell ergodicity states the equality between ensemble and time averages of observables.

Each measurement of an observable $f$ at time $t_0$ takes a certain interval of time to be performed. In this period of time the observable $f$ samples different values so that the effectively measured quantity is the time average
\begin{equation}
\frac{1}{t}\int_{0}^{t} f(T_s(q_0,p_0))\d s,
\end{equation}
where, $(q_0,p_0)$ is the microstate at $t=0$, and $\{T_s\}_{s\in\mathbb{R}}$ is the Hamiltonian flow generated by $H$, as given in equation~(\ref{hamilt}). Thus we are sampling $f$ on the trajectory whose initial point is $(q_0,p_0)$. Moreover, since the time interval during the measurement is very large compared to microscopic time scales, it is legitimate to take the limit $t\to\infty$:
\begin{equation}
f^*(q_0,p_0)=\lim_{t\to\infty}\frac{1}{t}\int_{0}^{t} f(T_s(q_0,p_0))\d s.
\label{eq:timeave}
\end{equation}
The ergodic problem questions when (and if) it may happen that ensemble averages~(\ref{ensavg}) equal time averages~(\ref{eq:timeave}). In general the answer to this question is negative, since $f^*(q_0,p_0)$ depends on the initial condition chosen on the trajectory, while $\langle f\rangle$ does not. Moreover 
$\langle f\rangle$ may depend on time (as long as $\rho$ does), while $f^*(q_0,p_0)$ does not. Notwithstanding these discrepancies, there may be some cases when the equality between time and ensemble averages holds. In the case of statistical equilibrium, say $\rho$ is time independent, these difficulties vanish.

As long as the ergodic hypothesis is true, then
\begin{equation}
f^*(q_0,p_0)= \langle f\rangle_{\mathrm{mc}},
\end{equation}
where $\langle f\rangle_{\mathrm{mc}}$ is the average of $f$ in the microcanonical ensemble. 
Recall that the ergodic hypothesis states that the trajectory of the system in $\Gamma$, the phase space, samples a dense subset of the whole energy shell~\eqref{eq:energy shell}. 
For further readings on the ergodic hypothesis we  recommend~\cite{oliveira,prentis,gallavotti}

\subsection{The Quantum Postulate of Equal a Priori Probability}

The above discussion is completely classical. 
 Since our world is quantum, a quantum formulation of statistical mechanics is in need. 
 Instead of probability distributions on phase space, one should consider density matrices, which encode the whole physical content of the system. Then, one formulates a quantum version of the postulate of equal a priori probability, as we are going to state.

Assume that the system  is described in a Hilbert space $\mathcal{H}$ and its evolution is generated by a Hamiltonian operator $H$. Next, fix a small energy shell around the value $E$, say $[E,E+\delta]$, where $\delta\ll E$ (on a macroscopic scale), but $\delta$ is large enough so that the shell contains many eigenvalues of $H$.

Instead of the phase space region $\Gamma_E$ of~\eqref{eq:microens2}, we have to consider now the spectral subspace of $H$ on  $[E,E+\delta]$, that is the subspace spanned by all eigenvectors with energy eigenvalues belonging to the energy shell, denoted by
\begin{equation}
\mathcal{H}_R=\mathcal{H}_{[E,E+\delta]}.
\label{eq:spectralsub}
\end{equation}
By the assumptions on $\delta$, one gets that $d_R= \dim \mathcal{H}_R \gg 1$.

The postulate of equal a priori probability affirms that all the states compatible with the energy $E$, that is all energy eigenvectors belonging to $\mathcal{H}_R$, are  equiprobable. But in the quantum world this is not enough. In fact, these states must be in an incoherent superposition. This is the content of the \textit{random phase postulate} and it is a purely quantum contribution to the foundations of statistical mechanics. The idea is that the  Fourier coefficient of a vector state in $\mathcal{H}_R$ should have equal probability and completely random phases, due to the unavoidable interactions between the environment and the system \cite{Huang}.

From these two postulates it follows that the state of the universe is described in terms of 
\begin{equation}
P_R = \chi_{[E, E+ \delta]} \big(H\big),
\end{equation}
the spectral projection of $H$ on $\mathcal{H}_R$. The microcanonical ensemble (density matrix) is the equiprobable state on the restriction of $\mathcal{H}_R$:
\begin{equation}
\mathcal{E}_R= \frac{P_{R}}{d_R}, 
\label{eq:microcandef}
\end{equation}
where $d_R = \dim \mathcal{H}_R = \tr P_R$.
Equivalently, the quantum postulate of equal a priori probability implies that the system is in a totally mixed state.

In the next lecture we are going to show how the postulate of (apparent) equal a priori probability can be proved, rather than postulated, from the very structure of quantum mechanics, and in particular from entanglement. But first we want to understand what is the state of a small subsystem of a body in a microcanonical ensemble.

\subsection{The Canonical Ensemble}

\label{sec:canonicalensemble}

We will show that, under some quite general assumptions, if a system is in a microcanonical state~\eqref{eq:microcandef}, every small part of it will be in a canonical state, $\Omega_S \propto \e^{-\beta H_S}$, characterized by a Boltzmann distribution among its eigenstates at a given temperature $\beta^{-1}$. 

Let us split the global system under consideration (call it the \emph{universe}) in two part: a (sub)\emph{system} $S$ and a \emph{bath} $B$.
The Hilbert space describing this composite system is 
\begin{equation}
\mathcal{H}=\mathcal{H}_S\otimes\mathcal{H}_B. 
\end{equation}
The Hamiltonian operator is accordingly split as 
\begin{equation}
\label{eqhamilt}
H= H_S\otimes\mathbb{I}_B+\mathbb{I}_S\otimes H_B + H_{\mathrm{\mathrm{int}}},
\end{equation}
where the Hamiltonians $H_S$ and $H_B$ act separately on the system and the bath, respectively, $\mathbb{I}_S$ and $\mathbb{I}_B$ are respectively the identity operators on $\mathcal{H}_S$ and $\mathcal{H}_B$, and $H_{\mathrm{int}}$ describes the interaction  between the system and the bath. Let us suppose that the interaction is very weak, that is 
\begin{equation}
\|H_{\mathrm{int}}\| \ll \|H_{S}\|,\|H_{B}\|.
\label{eq:smallint}
\end{equation}
Moreover, assume that the dimension of $\mathcal{H}_B$  is much larger than the dimension  of $\mathcal{H}_S$, namely
\begin{equation}
d_B=\dim\mathcal{H}_B \gg d_S=\dim\mathcal{H}_S.
\end{equation}

Next, fix a small energy shell around the macroscopic energy value $E$, say $[E,E+\delta]$, where $\delta\ll E$ on macroscopic scales, large enough to contain many eigenvalues of $H_B$, and let $P_R$ be the projection  
on the spectral subspace of $H$ on the energy shell~\eqref{eq:spectralsub}.
Assume that the universe is in the microcanonical state $\mathcal{E}_R$ given by~\eqref{eq:microcandef}.

The state of the system $S$ can be obtained by partial tracing the state of the universe 
over the bath, namely,
\begin{equation}
\Omega_S= \tr_B\,\mathcal{E}_R.
\end{equation}
By  following~\cite{landau}, we want to show that $\Omega_S$ is a thermal state at a given temperature. 

Let $\{|E_k\rangle\}_{k=1}^{d_B} \subset \mathcal{H}_B$ and $\{|\varepsilon_\alpha\rangle\}_{\alpha=1}^{d_S}\subset \mathcal{H}_S$ be the energy eigenstates of $H_B$ and $H_S$ respectively. Thus, 
\begin{equation}
\{|\varepsilon_\alpha\rangle \otimes |E_k\rangle \,:\,   1\leq \alpha\leq d_S,  1\leq k\leq d_B  \}
\end{equation}
is a basis of the Hilbert space of the universe $\mathcal{H}$. In view of condition~\eqref{eq:smallint}, this basis is an approximate eigenbasis of the total Hamiltonian with eigenvalues~$\approx \varepsilon_\alpha+ E_k$. Therefore, in terms of this basis the equiprobable state $\mathcal{E}_R$ reads
\begin{equation}
\mathcal{E}_R = \frac{P_R}{d_R} \approx \frac{1}{d_R} \sum_{\alpha,k} \chi_{[E,E+\delta]} (\varepsilon_\alpha+ E_k) \,
|\varepsilon_\alpha\rangle \langle \varepsilon_\alpha| \otimes |E_k\rangle \langle E_k|.
\end{equation}
Notice that the  sum is restricted to indexes $k$ and $\alpha$ such that $\varepsilon_\alpha+ E_k \in [E, E+\delta]$, where the characteristic function $\chi$ does not vanish. 

By tracing over the bath one gets
\begin{equation}
\Omega_S = \tr_B \mathcal{E}_R = \frac{1}{d_R} \sum_{\alpha,k} \chi_{[E,E+\delta]} (\varepsilon_\alpha+ E_k) \,
 |\varepsilon_\alpha\rangle \langle \varepsilon_\alpha|
= \frac{1}{d_R} \sum_{\alpha} d_\alpha^{(B)} |\varepsilon_\alpha\rangle \langle \varepsilon_\alpha|,
\end{equation}
where 
\begin{equation}
d_\alpha^{(B)} = \sum_{k} \chi_{[E,E+\delta]} (\varepsilon_\alpha+ E_k) = \sum_{k} \chi_{[E - \varepsilon_\alpha, E- \varepsilon_\alpha +\delta]} (E_k).
\end{equation}
Since $H_B = \sum_k E_k |E_k\rangle \langle E_k|$,
we get that
\begin{equation}
d_\alpha^{(B)} = \tr \chi_{[E - \varepsilon_\alpha, E- \varepsilon_\alpha +\delta]} (H_B) = \dim \mathcal{H}_{[E - \varepsilon_\alpha, E- \varepsilon_\alpha +\delta]}^{(B)}, 
\end{equation}
where $\mathcal{H}_{[E_1, E_2]}^{(B)}\subset \mathcal{H}_B$ is the spectral subspace of $H_B$ on $[E_1,E_2]$, that is the subspace generated by all eigenvectors with energy in $[E_1,E_2]$. Thus $d_{\alpha}^{(B)}$ is a nonnegative integer, and it may vanish!

Let us define the bath entropy at energy $E$ as~\cite{landau}
\begin{equation}
S_B(E) = \ln \dim \mathcal{H}_{[E, E+\delta]}^{(B)},
\end{equation}
that is the logarithm of the number of bath energy levels in the energy shell $[E, E+\delta]$: This is Boltzmann's statistical entropy.
We get that
\begin{equation}
d_\alpha^{(B)} =  \dim \mathcal{H}_{[E - \varepsilon_\alpha, E- \varepsilon_\alpha +\delta]}^{(B)} = \e^{S_B (E- \varepsilon_\alpha )}.
\end{equation}

Since the dimension of $\mathcal{H}_B$ is very large, $d_R\gg1$, we can assume that the spectrum of $H_B$ is quasi-continuous, so that $S_B(E)$ can be considered a continuous differentiable function of $E$. Thus, by assuming that the microscopic energy $\varepsilon_\alpha$ is much smaller than the macroscopic energy $E$, i.e.\ $\varepsilon_\alpha \ll E$,  we can write
\begin{equation}
S_B(E- \varepsilon_\alpha ) \approx S_B(E) -\frac{\d S_B (E)}{\d E} \, \varepsilon_\alpha.
\end{equation}
Therefore, we get
\begin{equation}
\label{canoneq}
\Omega_S= \frac{1}{d_R} \sum_{\alpha} d_\alpha^{(B)} |\varepsilon_\alpha\rangle \langle \varepsilon_\alpha|
\approx \frac{1}{Z}  \sum_{\alpha}  \e^{-\beta \varepsilon_\alpha} |\varepsilon_\alpha\rangle \langle \varepsilon_\alpha| =
\frac{1}{Z}  \e^{-\beta H_S},
\end{equation}
where $Z= \tr \e^{-\beta H_S}$, and 
\begin{equation}
\beta = \frac{\d S_B (E)}{\d E}.
\end{equation}
is the thermodynamical expression of the inverse temperature of the bath.

\section{Lecture 2: Entanglement and the foundations of Statistical Mechanics}

In this lecture we are going to show how the postulate of equal a priori probability can be proved, rather than postulated, from the very structure of quantum mechanics, and in particular from entanglement.

We will start by recalling some basic facts about entanglement, which according to Schr\"odinger is \textit{the characteristic trait of quantum mechanics}~\cite{schr}.

Let us consider a composite system of two spins on the Hilbert space $\mathcal{H}=\mathcal{H}_S\otimes\mathcal{H}_B$, where $\mathcal{H}_S=\mathcal{H}_B= \mathbb{C}^2.$ As before, the subscript $S$ will stand for \emph{system} while $B$ will stand for  \emph{bath} or \emph{environment}. The system together with its environment will be called \emph{universe}. Moreover we denote by $\{|\up\rangle,|\down\rangle\}$ the computational basis of $\mathbb{C}^2$~\cite{nielsen}.

If the bipartite system is described in terms of a factorized state, for example,
\begin{equation}
|\phi\rangle=|\up\rangle_S\otimes|\up\rangle_B
\end{equation}
then the system and the enviroment can be described independently. On the other hand, when the global state is not factorized, for example the Bell state:
\begin{equation}
|\Phi^+\rangle=\frac{1}{\sqrt{2}}\left(|\up\rangle_S\otimes|\up\rangle_B+|\down\rangle_S\otimes|\down\rangle_B\right),
\end{equation}
then the state is entangled. In fact, in this case the system is described by the density matrix
\begin{equation}
\label{eq:mixed}
\rho_S=\tr_B (|\Phi^+\rangle\langle\Phi^+|)= \frac{1}{2}\,\mathbb{I}_2,
\end{equation}
where $\mathbb{I}_2$ is the identity operator on $\mathcal{H}_S=\mathbb{C}^2$.
Equation~(\ref{eq:mixed}) tells us that the system $S$ is in a totally mixed state, that is, it is randomly distributed.

More generally, we can consider the family of states in $\mathcal{H}$
\begin{equation}
|\Phi_\alpha\rangle = \sqrt{\alpha} |\up\rangle_S\otimes|\up\rangle_B+\sqrt{1-\alpha}|\down\rangle_S\otimes|\down\rangle_B, \qquad \alpha \in[0,1],
\end{equation}
which embeds both the separable state, $|\phi\rangle = |\Phi_{0}\rangle$ for $\alpha=0$, and the Bell state $|\Phi^+\rangle = |\Phi_{1/2}\rangle$ for $\alpha=1/2$. The reduced density matrix of the system reads
\begin{equation}
\rho_{S}^\alpha = \alpha |\up\rangle\langle\up|+ (1-\alpha)|\down\rangle\langle\down|,
\end{equation}
which is a state whose mixture depends on $\alpha$.

Summing up, if the state of the universe is factorized the information on the whole state and on every subsystem is completely accessible. On the contrary, if the global state is entangled, notwithstanding one has a complete knowledge of the state of the universe, a priori only a partial  knowledge on the subsystem can be obtained. When the global state is maximally entangled ($\alpha=1/2$) one has no information at all on the subsystem.

Mathematically speaking the information content of a state $\rho$ is described by the von Neumann entropy~\cite{vn}:
\begin{equation}
S(\rho)=-\tr(\rho\ln\rho)=-\sum_k p_k \ln p_k
\end{equation}
where $p_k$ are the eigenvalues of $\rho$, $0\le p_k \le 1$ and $\sum_k p_k=1$.
 
It is easy to see that every pure state has 0 entropy, which means that the information encoded in the state is completely available. In our case, for example, $S(|\Phi_\alpha\rangle\langle\Phi_\alpha|)=0$ for all $\alpha$. 
On the other hand, one gets
\begin{equation}
S(\rho_{S}^\alpha)=-\alpha \log\alpha -(1-\alpha) \log(1-\alpha),
\label{eq:Shan}
\end{equation}
that is a positive symmetric function for $\alpha\in[0,1]$, which is $0$ for separable global states ($\alpha=0$ and $\alpha=1$) and reaches its maximum $S=\ln 2= \ln d_S$ for the maximally entangled Bell state ($\alpha=1/2$). In this latter case the entropy is maximal and it corresponds to a complete ignorance on the subsystem $S$. Notice that~\eqref{eq:Shan} is nothing but the Shannon entropy~\cite{shannon} of the probability vector $(\alpha, 1-\alpha)$.

In general, given a pure state $|\Psi\rangle$ of a composite system $\mathcal{H}=\mathcal{H}_S\otimes\mathcal{H}_B$ with generic dimensions $d_S = \dim \mathcal{H}_S \leq d_B=\dim \mathcal{H}_B$, one gets that
\begin{equation}
0\leq S(\rho_S) \leq \ln d_S.
\end{equation}
Here, $S(\rho_S) =0$ for separable states, $|\Psi\rangle= |u\rangle\otimes|v\rangle$, while $S(\rho_S) =\ln d_S$ for maximally entangled states, \begin{equation}
|\Psi\rangle= \frac{1}{\sqrt{d_S}} \sum_{k=1}^{d_S }|u_k\rangle\otimes|v_k\rangle,
\end{equation} 
with $\{u_k\}$ and $\{v_k\}$ being orthonormal systems.

The von Neumann entropy is a measure of entanglement which leads to an objective lack of knowledge. In fact, even if we had complete information on the state of the universe (i.e. it is in a pure state and has zero entropy), the state of any subsystem could be mixed and have nonzero entropy, and, as such, it would behave like a probability distribution over pure states.
This is manifestly a purely quantum phenomenon, since no counterpart exists in classical mechanics. Classically, in fact, the complete knowledge of the state of the universe implies a complete knowledge of the state of any subsystem.

In the following we will show that  almost all pure states of a composite system with $d_B\gg d_S$ are highly entangled, and thus the system $S$ is typically in a highly mixed state.  More precisely, we will prove ``canonical typicality'', which mantains that the system will be thermalized (that is, in the canonical state) for almost all pure states of the universe.  Therefore, the postulate of equal a priori probability, which refers to ensembles or time averages of  states of the universe, and as such relies on a subjective lack of information, can be dismissed and one can refer only to pure states of the universe. The lack of information which will give a canonical density matrix for the system is just a physical consequence of entanglement between the system and its environment.

\subsection{Canonical typicality}
In this section we would like to show that the principle of equal a priori probability, which cannot be proved, should be replaced by the principle of \emph{Canonical Typicality}, which is based on individual states rather than ensembles or time averages and, most importantly, can be proved. This principle was named this way by~\cite{goldstein}, and is also known under the name \emph{Quantum Typicality}~\cite{gemmer} or \emph{General Canonical Principle}~\cite{popescun}. 

In this new approach thermalization emerges as a consequence of entanglement between a system and its environment. 
This idea goes back to Schr\"odinger (see the Appendix in~\cite{schr1}) and to von Neumann in his formulation of the quantum ergodic theorem~\cite{vn1}. Then it reappeared several times up to today~\cite{gemmer,goldstein,popescu,popescun}. 

In our deduction of  canonical typicality we are going to follow~\cite{popescu}.

Suppose the universe has to obey some global constraint, say \textit{R}, which translates into the choice of a subspace of the total Hilbert space, say
\begin{equation}
\mathcal{H}_R\subset\mathcal{H}_S\otimes\mathcal{H}_B.
\end{equation}
As before we are going to denote the dimensions of $\mathcal{H}_S$,$\mathcal{H}_E$ and $\mathcal{H}_R$, respectively, by $d_S$, $d_E$ and $d_R$. In the standard approach to statistical mechanics, as seen in the previous section, the restriction $R$ is imposed on the total energy of the universe. In this case, however, we let the restriction be completely arbitrary.

Moreover, the \emph{equiprobable state} in $\mathcal{H}_R$ is denoted by
\begin{equation}
\mathcal{E}_R= \frac{P_{R}}{d_R},
\label{eq:equiprobableR}
\end{equation}
where $P_{R}$ is the projection on $\mathcal{H}_R$. In this case equal probabilities (and random phases) are assigned to all the states of the universe which are consistent with the constraint $R$. When the latter is imposed on the total energy of the universe, $\mathcal{E}_R$ is nothing but the microcanonical state considered in the previous lecture.

The \emph{(generalized) canonical state} of system $S$ is defined as the trace over the bath of $\mathcal{E}_R$, that is:
\begin{equation}
\Omega_S=\tr_B \mathcal{E}_R.
\label{eq:OmegaS}
\end{equation}

Instead of considering the universe in the equiprobable state $\mathcal{E}_R$, which describes subjective ignorance, we will consider it in a pure state $|\phi\rangle$ in $\mathcal{H}_R$, such that $\langle\phi|\phi\rangle=1$. In such a case  the system is described by the density matrix 
\begin{equation}
\label{eq:rhoS}
\rho_S=\tr_B( |\phi\rangle\langle\phi|).
\end{equation} 
The question is to understand how much different is $\rho_S$ from the canonical state $\Omega_S$.
The answer is provided by a theorem given in~\cite{popescu}, which states that $\rho_S$ is almost equal to $\Omega_S$ for almost every pure state compatible with the constraint $R$. 

From this theorem  canonical typicality follows:
\vspace{0.3cm}
\\
\emph{Given a sufficiently small subsystem of the universe, a typical pure state of the universe is such that the subsystem is approximately in the canonical state $\Omega_S$.}
\vspace{0.3cm}

This means that for almost every state $|\phi\rangle \in \mathcal{H}_R$ of the universe, the system behaves as if the universe were in the equiprobable state $\mathcal{E}_R$. Thus, the state of the universe is locally (on the system $S$) practically indistinguishable from $\mathcal{E}_R$.

Moreover, it is important to stress that $\Omega_S$ is not necessarily the thermal canonical state~(\ref{canoneq}), but rather a (generalized) canonical state with respect to the arbitrary restriction $R$ chosen.
Of course, if $R$ is a restriction on the total energy as in~\eqref{eq:spectralsub}
and under the conditions on the total Hamiltonian $H$ considered in the previous lecture (Sec.~\ref{sec:canonicalensemble}) almost every pure state $|\phi\rangle$ of the universe is such that the system $S$ is approximately in the canonical thermal state $\e^{-\beta H_S}/Z$, as in equation~(\ref{canoneq}).

Thus there is a link between  canonical typicality and the standard approach to statistical mechanics.
Yet the core of  canonical typicality does not lie in the explicit expression of $\Omega_S$, which is a standard problem in statistical mechanics and depends on the structure of a given Hamiltonian $H$, but only in the equality 
\begin{equation}
\rho_S\approx\Omega_S,
\label{eq:cantypqual}
\end{equation}
which is of a  purely \emph{kinematic} nature.
It may happen, for example, that for a strongly long-range interacting system the interaction Hamiltonian in~\eqref{eqhamilt} is not negligible, so that the canonical state cannot have the expression~(\ref{canoneq}), and the very concept of temperature is questionable,  but Eq.~\eqref{eq:cantypqual} still holds.

\subsection{Quantitative arguments}

In order to be more quantitative it is essential to explain what the vague expressions like  ``sufficienly small subsystem'',  
``approximately in the canonical state'', and ``a typical pure state'' mean. In particular, we need to define a \emph{distance} between states $\rho_S$ and $\Omega_S$ and a \emph{measure} over the pure states $|\phi\rangle$ with respect to which typicality is defined.

\subsubsection{Distance}
As a distance between $\rho_S$ and the canonical state $\Omega_S$ we will use the trace distance, $\|\rho_S-\Omega_S\|_1$, which is induced by the trace norm
\begin{equation}
\label{eq:tracenorm}
\|\rho\|_1=\tr |\rho| = \tr \sqrt{\rho^\dagger \rho}.
\end{equation}
This distance represents (two times) the maximal difference in the probability of obtaining any outcome for any measurement performed on the two states $\rho_S$ and $\Omega_S$. Indeed, since by duality
\begin{equation}
\|\rho\|_1 = \sup_{\|M\|=1} |\tr(\rho M)|,
\label{eq:duality}
\end{equation}
we get that the difference of the expectation values of an observable $M$ in the two states satisfies the inequality
\begin{equation}
|\tr(\rho_S M) -\tr(\Omega_S M) |
\leq \|\rho_S-\Omega_S\|_1\|M\|.
\end{equation}
Thus, the trace distance quantifies how hard is to tell $\rho_S$ and $\Omega_S$ apart by means of quantum measurements $M$.

A distance easier to handle is that induced by the Hilbert-Schmidt norm
\begin{equation}
\label{eq:HSnorm}
\|\rho\|_2=\sqrt{ \tr (\rho^\dagger \rho) },
\end{equation}
where the square root is taken after the trace. It is easy to prove that $\|\rho\|_2 \leq \|\rho\|_1 \leq  \sqrt{d} \|\rho\|_2$ with $d$ being the dimension of the Hilbert space. However, the Hilbert-Schmidt distance has not a nice operational meaning like the trace distance, and in fact in higher dimension can be very small even if the two states have disjoint supports.
\begin{example}
Consider in $\mathbb{C}^{2d}$ the two states $\rho_1 = P_1/d$ and $\rho_2 = (1-P_1)/d$, where $P_1$ is a rank-$d$ projection. Notice that they have disjoint supports. By a straightforward computation  one gets that
\begin{equation}
\|\rho_1-\rho_2\|_1 = 2, \qquad \|\rho_1-\rho_2\|_2 = \sqrt{\frac{2}{d}},
\end{equation}
so that the norm distance between $\rho_1$ and $\rho_2$ is constant and maximal, while  the Hilbert-Schmidt distance becomes arbitrarily small as
$d$ increases. 
\end{example}

\subsubsection{The uniform measure on  pure states}

Let $|\phi\rangle$ be a pure state in $\mathcal{H}_R$. Due to the normalization condition, $\langle\phi|\phi\rangle=1$,  $|\phi\rangle$~belongs to the unit sphere of $\mathcal{H}_R$.
Indeed, let $\{|u_k\rangle\}_{k=1}^{d_R}$  be an orthonormal basis of $\mathcal{H}_R$. The unit vector $|\phi\rangle$ admits a unique decomposition
\begin{equation}
|\phi\rangle=\sum_{k=1}^{d_R}z_k\,|u_k\rangle, 
\end{equation}
in terms of its Fourier coefficients $z_k= \langle u_k|\phi\rangle\in \mathbb{C}$, for $k =1,\dots,  d_R$.

Consider now the normalization constraint $\langle\phi|\phi\rangle=1$, and the decomposition of $z_k$ into its real and imaginary parts, $z_k=x_k+ i y_k$, so that:
\begin{equation}
\langle\phi|\phi\rangle=\sum_{k=1}^{d_R} |z_k|^2=\sum_{k=1}^{d_R} x_k^2+\sum_{k=1}^{d_R} y_k^2=1.
\end{equation}
The latter equation tells us that $|\phi\rangle$ belongs to a $(2d_R-1)$-dimensional (real) sphere $\mathbb{S}^{2 d_R-1}\subset\mathcal{H}_R$. 

Let us consider the uniform probability measure on the sphere, say $\sigma(\mathbb{S}^{2 d_R-1})=1$. The measure $\sigma$ is rotationally invariant, that is unitarily invariant in $\mathcal{H}_R$, and the expectation value of a function on the sphere is given by:
\begin{equation}
\big\langle f(|\phi\rangle)\big\rangle=\int_{\mathbb{S}^{2 d_R-1}}f(|\phi\rangle)\d\sigma.
\end{equation}
First of all we note that $\big\langle |\phi\rangle \big\rangle=0$. In fact the state $|\phi\rangle$ is uniformly distributed on the sphere and for this reason $\langle z_k\rangle=0$, for every $k$. Moreover:
\begin{equation}
1=\big\langle \|\phi\|^2 \big\rangle=\big\langle \sum_{k=1}^{d_R} |z_k|^2 \big\rangle=\sum_{k=1}^{d_R} \big\langle |z_k|^2 \big\rangle.
\end{equation}
Due to rotationally invariance it follows that $\big\langle |z_k|^2 \big\rangle$ is independent of $k$, and thus $\big\langle |z_k|^2 \big\rangle=1/d_R$. 

\subsubsection{Average vs Typical}

If we compute the average $\big\langle |\phi\rangle\langle\phi| \big\rangle$, we get
\begin{equation}
\big\langle |\phi\rangle\langle\phi| \big\rangle= \big\langle \sum_{k,l=1}^{d_R}z_k \overline{z}_l|u_k\rangle\langle u_l| \big\rangle=\sum_{k,l=1}^{d_R}\big\langle z_k \overline{z}_l
\big\rangle|u_k\rangle\langle u_l|=\frac{1}{d_R}\sum_{k=1}^{d_R}|u_k\rangle\langle u_k|=\frac{P_R}{d_R},
\end{equation}
where we used the fact that $\big\langle z_k \overline{z}_l\big\rangle= \delta_{k,l}/d_R $. Indeed,
\begin{equation}
\big\langle z_k \overline{z}_l\big\rangle= \big\langle x_k  x_l\big \rangle+
\big \langle y_k  y_l\big \rangle
-i\big\langle x_k y_l\big \rangle +i \big\langle x_l y_k\big\rangle ,
\end{equation}
and $\big\langle x_k x_l\big\rangle=\big\langle y_k y_l\big\rangle=\delta_{k,l}/(2 d_R)$, while $\big\langle x_k y_l\big\rangle=\big\langle x_l y_k\big\rangle=0$.
Therefore, we get that the the equiprobable state is nothing but the average state of the universe in $\mathcal{H}_R$:
\begin{equation}
\mathcal{E}_R=\big\langle |\phi\rangle\langle\phi| \big\rangle.
\label{eq:ER=avg}
\end{equation}
By taking the partial trace over the bath of both sides of~(\ref{eq:ER=avg}) we immediately get
\begin{equation}
\Omega_S=\big\langle\rho_S\big\rangle,
\label{eq:rhosavg}
\end{equation}
where we used the fact that $\big\langle \tr_B |\phi\rangle\langle\phi|\big\rangle =\tr_B  \big\langle |\phi\rangle\langle\phi|\big\rangle$ and definitions~(\ref{eq:OmegaS}) and~(\ref{eq:rhoS}).

Equation~\eqref{eq:rhosavg} tells us that the average state of the system is the canonical state $\Omega_S$. In other words, \emph{on average} the system reduced state  of a pure state of the universe $|\phi\rangle$ (constrained to $\H_R$) is the canonical state $\Omega_S$, that is the system reduced state of the equiprobable state of the universe $\mathcal{E}_R$: on average one cannot distinguish locally whether the universe is in a  pure state or in the maximally mixed state.

However, this is not enough: the \emph{average} behavior may give a very loose information on the behavior of single individuals, and even on the \emph{typical} behavior, that is the behavior of a large multitude (see Schr\"odinger's quote at the beginning of the Introduction, Sec.~\ref{sec:intro}).
In fact,  it may happen that a large part of the available states could be far apart from the average. 

As a simple example consider a macroscopic system made up of spins which can assume only the values $\pm1$. Furthermore, suppose that half of them are $+1$ and the other half are $-1$, so that the average spin equals 0. In this situation the average by itself has no physical content, inasmuch as there is not even a single actual spin with the average feature!

What really matters for a typical behavior are also the fluctuations around the average and the possibility to control them; in fact, when the fluctuations (that is the variance) are very small, then the average becomes a physically relevant parameter, since the large majority exhibits a behavior which is very close to the average one.

Therefore we are going to look now at the fluctuations around the average and to prove that
\begin{equation}
\big\langle\|\rho_S - \Omega_S\|_1\big\rangle \leq  \sqrt{ \frac{d_S^2}{d_R} },
\label{eq:fluctuations}
\end{equation}
so that under the sole condition $d_R \gg d_S^2$ the fluctuations around the average  $\Omega_S=\big\langle\rho_S\big\rangle$ are negligible, and canonical typicality~\eqref{eq:cantypqual} holds \cite{gemmer,goldstein}. 

In fact, by using L\'evy's lemma, a profound result in convex geometry, Popescu, Short and Winter~\cite{popescu,popescun} have proved that inequality~\eqref{eq:fluctuations} implies that~\eqref{eq:rhosavg} is true for the overwhelming majority of pure states $|\phi\rangle$, and does not hold on a set of pure states $|\phi\rangle$ exponentially small in $d_R$. This is the content of the following theorem that we are going to discuss.

\begin{theorem}[Canonical Typicality \cite{popescu}]
\label{thm:canonicaltypicality}
For a randomly chosen state $|\phi\rangle\in\mathcal{H}_R\subset\mathcal{H}_S\otimes\mathcal{H}_B$ and arbitrary $\varepsilon > 0$ the distance between the reduced density matrix $\rho_S={\tr_B}\left(|\phi\rangle\langle\phi|\right)$ and the canonical state $\Omega_S={\tr_B}\mathcal{E}_R$ is given probabilistically by:
\begin{equation}
\label{eqtyp}
\mathrm{Prob}\big(\|\rho_S-\Omega_S\|_1\ge\eta\big)\le\eta',
\end{equation}
where
\begin{equation}
\eta=\varepsilon+\sqrt{\frac{d_S}{d_B^{\mathrm{eff}}}},\qquad\eta'=2\exp\left(-Cd_R\varepsilon^2\right),
\end{equation}
with
\begin{equation}
\label{const}
C=\frac{1}{18\pi^3},\qquad d_B^\mathrm{eff}=\frac{1}{\tr{\Omega_B^2}},\qquad\Omega_B=\tr_S\mathcal{E}_R,
\end{equation}
and $d_S=\dim\mathcal{H}_S$, $d_R=\dim\mathcal{H}_R$.
Moreover, it results that 
\begin{equation}
d_B^{\mathrm{eff}}\ge \frac{d_R}{d_S} .
\label{eq:effmin}
\end{equation}
\end{theorem}
We observe that when $\eta$ and $\eta'$ are small enough the state $\rho_S$ will be sufficiently close to the canonical state $\Omega_S$, with high probability. For small $\varepsilon$  this happens as long as the effective dimension of the environment, $d_B^\mathrm{eff}$, is much larger than the dimension of the system $d_S$, and the dimension of the accessible space $d_R$ is much larger than $\varepsilon^{-2}$.  

Notice that from~(\ref{eq:effmin}) one gets that
\begin{equation}
\eta \leq \varepsilon + \sqrt{\frac{d_S^2}{d_R}}=\tilde\eta,
\end{equation}
so that 
\begin{equation}
\mathrm{Prob}\big(\,\|\rho_S-\Omega_S\|_1\ge\tilde\eta\,\big) \leq \mathrm{Prob}\big(\,\|\rho_S-\Omega_S\|_1\ge\eta\,\big)\le\eta',
\end{equation}
that is
\begin{equation}
\mathrm{Prob}\Bigg(\,\|\rho_S-\Omega_S\|_1\ge \varepsilon + \sqrt{\frac{d_S^2}{d_R}}\,\Bigg) 
\leq 2\exp\left(-Cd_R\varepsilon^2\right).
\end{equation}
For example, if the total accessible space is large ($d_R\gg d_S^2$), and one chooses $\varepsilon=d_R^{-1/3}$, then $\rho_S\approx\Omega_S$ for the overwhelming majority of pure states $|\phi\rangle$ of the universe. Indeed,
\begin{equation}
\mathrm{Prob}\left(\|\rho_S-\Omega_S\|_1\ge d_R^{-1/3} + d_S d_R^{-1/2}\right) 
\leq 2\exp\left(-Cd_R^{1/3}\right),
\end{equation}
and $\rho_S \to \Omega_S$ in probability as $d_R\to\infty$.

\subsection{Proof of Theorem~\ref{thm:canonicaltypicality}}
\label{sec:Levy}

A crucial ingredient in the proof of Theorem~\ref{thm:canonicaltypicality} is L\'evy's Lemma, which we briefly recall. Roughly speaking L\'evy's lemma states that the value of any regular function on a high dimensional sphere is almost everywhere equal to its  average value. More precisely:
\begin{lemma}[L\'evy]
\label{lem:Levy}
Let $f:\mathbb{S}^n\to\mathbb{R}$ be a continuous function on the $n$-dimensional sphere $\mathbb{S}^n$ with Lipschitz constant $\eta$. Let $\phi$ be a point on the sphere chosen uniformly at random, then for all $\varepsilon>0$:
\begin{equation}
\mathrm{Prob}\left(|f(\phi)-\big\langle f(\phi)\big\rangle|\ge\varepsilon\right)\le 2 \exp\left(-\frac{2 C (n+1)}{\eta^2}\varepsilon^2\right)
\label{eq:Levy}
\end{equation}
where C is given in equation \eqref{const}.
\end{lemma}
This means that the set of exceptional points, where the value of the function differs appreciably (i.e. more than $\varepsilon$) from its average value,  is exponentially small.
Recall that the Lipschitz constant of $f$ is the minimum $c>0$ such that 
\begin{equation}
|f(\phi_1) - f(\phi_2)|\leq c |\phi_1-\phi_2|
\end{equation} 
for all $\phi_1,\phi_2\in \mathbb{S}^n$. In particular if $f$ is differentiable with bounded derivative, then $\eta= \max_{\phi} |f'(\phi)|$.

Let us now apply L\'evy's lemma in order to prove equation \eqref{eqtyp}. 
Define 
\begin{equation}
\label{eq:fdef}
f(|\phi\rangle)= \|\rho_S-\Omega_S\|_1,
\end{equation} 
with $\rho_S = \tr_B |\phi\rangle\langle\phi|$. 
Preliminarily we are going to prove that

\begin{lemma}
Let $\eta$ be the Lipschitz constant of the function $f$ defined in~\eqref{eq:fdef}. One gets
$\eta\le2$.
\end {lemma}
\begin{proof}
Fix two pure states, 
say $|\phi_1\rangle$ and $|\phi_2\rangle$, and the respective reduced density matrices $\rho_1=\tr_B\left(|\phi_1\rangle\langle\phi_1|\right)$ and $\rho_2=\tr_B\left(|\phi_2\rangle\langle\phi_2|\right)$.
Consider now:
\begin{eqnarray}
|f(|\phi_1\rangle)-f(|\phi_2\rangle)|^2&=&|\,\|\rho_1-\Omega_S\|_1-\|\rho_2-\Omega_S\|_1|^2\\
&\le&\|\rho_1-\rho_2\|_1^2\\
&=& \| \tr_B \left(|\phi_1\rangle\langle\phi_1|-|\phi_2\rangle\langle\phi_2|\right)\|_1^2\\
&\le& \|\,|\phi_1\rangle\langle\phi_1|-|\phi_2\rangle\langle\phi_2|\,\|_1^2.
\end{eqnarray}
The last inequality holds since partial tracing reduces trace norm. Indeed, from~(\ref{eq:duality}) one gets
\begin{equation}
\|\rho\|_1 = \sup_{\|A\|=1} |\tr(A\rho)| \geq \sup_{\|C\|=1} \big|\tr\big((C\otimes \mathbb{I}_R)\rho\big)\big| 
= \sup_{\|C\|=1} |\tr(C \tr_B \rho)| = \|\tr_B\rho\|_1,
\end{equation}
where the inequality follows since the supremum is taken on the smaller set of operators of the form $A= C\otimes \mathbb{I}_R$.

Furthermore we claim that
\begin{equation}
\|\,|\phi_1\rangle\langle\phi_1|-|\phi_2\rangle\langle\phi_2|\,\|_1^2= 4 \left(1-|\langle\phi_1|\phi_2\rangle|^2\right)\le4\|\,|\phi_1\rangle-|\phi_2\rangle\|^2,
\label{eq:claim2}
\end{equation} 
where the last inequality follows from $\Re \langle\phi_1|\phi_2\rangle \leq \langle\phi_1|\phi_2\rangle$.
Therefore, 
\begin{equation}
|f(|\phi_1\rangle)-f(|\phi_2\rangle)|\leq  2 \|\,|\phi_1\rangle-|\phi_2\rangle\|,
\end{equation}
for all  $|\phi_1\rangle$ and $|\phi_2\rangle$, whence $\eta\le2$.
\end{proof}

Let us now return to our main purpose and apply L\'evy's lemma to the function $f$ in~(\ref{eq:fdef}) with $\eta\le2$ and $n=2d_R-1$:
\begin{equation}
\mathrm{Prob}\left(\,\big|\|\rho_S-\Omega_S\|_1-\big\langle \|\rho_S-\Omega_S\|_1 \big\rangle \big| \ge\varepsilon\,\right)\le 2\exp\left(-\frac{4 C d_R}{\eta^2}\varepsilon^2\,\right)
\le 2\exp\left(-C d_R\varepsilon^2\right)
\end{equation} 
Moreover, the following inequality holds:
\begin{eqnarray}
& &\mathrm{Prob}\left(\,\|\rho_S-\Omega_S\|_1\ge\varepsilon+\big\langle\|\rho_S-\Omega_S\|_1\big\rangle\,\right)
\nonumber\\
& &\qquad\qquad\le\mathrm{Prob}\left(\,\big|\|\rho_S-\Omega_S\|_1-\big\langle \|\rho_S-\Omega_S\|_1 \big\rangle \big| \ge\varepsilon\,\right),
\end{eqnarray}
since  the probability on the left hand side is taken on a subset of  the probability on the right hand side.

We claim  that 
\begin{equation}
\big\langle\|\rho_S-\Omega_S\|_1\big\rangle\le\sqrt{\frac{d_S}{d_B^{\mathrm{eff}}}},
\label{eq:claim1}
\end{equation} 
where $d_B^\mathrm{eff}$ is defined in equation \eqref{const}. Thus,
\begin{equation}
\mathrm{Prob}\left(\,\|\rho_S-\Omega_S\|_1\ge\varepsilon+\sqrt{\frac{d_S}{d_B^{\mathrm{eff}}}}\,\right)
\le\mathrm{Prob}\left(\,\|\rho_S-\Omega_S\|_1\ge\varepsilon+\big\langle\|\rho_S-\Omega_S\|_1\big\rangle\,\right).
\end{equation}

So that, putting all the ingredients together and defining:
\begin{equation}
\eta=\varepsilon+\sqrt{\frac{d_S}{d_B^{\,\mathrm{eff}}}}\qquad\eta'=2\exp\left(-\,C\,d_R\varepsilon^2\right)
\end{equation}
equation \eqref{eqtyp} holds.

\subsubsection{Proof of claim~\eqref{eq:claim2}}
Let $|\phi_1\rangle$ and $|\phi_2\rangle$ be states in $\mathcal{H}_R$ and define the operator $A=|\phi_1\rangle\langle\phi_1|-|\phi_2\rangle\langle\phi_2|$. This operator acts non trivially only on $\mathcal{K}$, the linear span of $\phi_1$ and $\phi_2$, and is zero outside. Our aim is to compute $\|A\|_1=\tr|A|=\tr\sqrt{A^\dagger A}$.

Decompose $\phi_2$ along $\phi_1$ and its orthogonal complement in $\mathcal{K}$, say, $\phi_2=\alpha\phi_1+\beta\phi_1^\perp$, with $\alpha = \langle \phi_1|\phi_2\rangle$ and $|\alpha|^2+|\beta|^2=1$.
After a straightforward computation one gets
\begin{eqnarray}
A^\dagger A &=&|\phi_1\rangle\langle\phi_1|-|\alpha|^2|\phi_1\rangle\langle\phi_1|+|\beta|^2|\phi_1^\perp\rangle\langle\phi_1^\perp|
\nonumber\\
&=& |\beta|^2 |\phi_1\rangle\langle\phi_1|+|\beta|^2|\phi_1^\perp\rangle\langle\phi_1^\perp| = |\beta|^2 \mathbb{I}_{\mathcal{K}}.
\end{eqnarray}
Thus $|A|= |\beta| \mathbb{I}_{\mathcal{K}}$ and, since $\dim\mathcal{K}= 2$, one gets 
\begin{equation}
\tr|A|=2|\beta| = \sqrt{1-|\langle\phi_1|\phi_2\rangle|^2}.
\end{equation}

\subsubsection{Proof of claim~\eqref{eq:claim1}}
The trace norm can be bounded above by the Hilbert-Schmidt norm:
\begin{equation}
\|\rho\|_1 = \tr\sqrt{\rho^\dagger \rho} = d_S \tr \left(\mathcal{E}_S \sqrt{\rho^\dagger \rho} \right)
\leq d_S \sqrt{\tr \left(\mathcal{E}_S {\rho^\dagger \rho} \right)} = \sqrt{d_S \tr \left(\rho^\dagger \rho \right)} =
\sqrt{d_S} \|\rho\|_2,
\label{eq:trN-HS}
\end{equation}
where $\mathcal{E}_S= \mathbb{I}_S/d_S$ is the equiprobable (microcanonical) state of $\mathcal{H}_S$, and the inequality follows from the concavity of the square root function.
We get that
\begin{equation}
\big\langle\|\rho_S - \Omega_S\|_2\big\rangle = \big\langle \sqrt{\tr(\rho_S - \Omega_S)^2}\big\rangle
\leq  \sqrt{\big\langle \tr(\rho_S - \Omega_S)^2 \big\rangle} =
\sqrt{ \tr \big\langle(\rho_S - \Omega_S)^2 \big\rangle}
\end{equation}
Now, 
\begin{equation}
\big\langle(\rho_S - \Omega_S)^2 \big\rangle = \big\langle\rho_S^2 \big\rangle - \Omega_S^2,
\end{equation}
since $\Omega_S = \big\langle\rho_S \big\rangle$ is the average reduced state.
Therefore,
\begin{equation}
\big\langle\|\rho_S - \Omega_S\|_2\big\rangle \leq  \sqrt{ \tr (\big\langle\rho_S^2 \big\rangle- \Omega_S^2) } .
\label{eq:71}
\end{equation}
The  standard deviation on the right hand side can be bounded above by~\cite{popescu}
\begin{equation}
\tr (\big\langle\rho_S^2 \big\rangle- \Omega_S^2) = \tr \big\langle\rho_S^2 \big\rangle - \tr \big\langle\rho_S \big\rangle^2 \leq \tr \big\langle\rho_B \big\rangle^2 ,
\end{equation}
where $\rho_B = \tr_S |\phi\rangle\langle\phi|$ is the reduced density matrix of the bath.
Notice that, from~\eqref{eq:ER=avg}, its average 
\begin{equation}
\big\langle\rho_B \big\rangle = \tr_S \big\langle|\phi\rangle\langle\phi| \big\rangle = \tr_S \mathcal{E}_R = \Omega_B,
\label{eq:rhoBavg}
\end{equation}
is nothing but the bath reduced state of the equiprobable state $\mathcal{E}_R$, in complete symmetry with the relation~\eqref{eq:rhosavg} for the system.

Finally, by gathering up~\eqref{eq:trN-HS}, \eqref{eq:71}-\eqref{eq:rhoBavg}, we get
\begin{equation}
\big\langle\|\rho_S - \Omega_S\|_1\big\rangle \leq  \sqrt{ d_S \tr \Omega_B^2 },
\end{equation}
which, by using definition~\eqref{const}, $d_B^\mathrm{eff}=1/\tr{\Omega_B^2}$, yields  claim~\eqref{eq:claim1}. 

\subsubsection{Proof of inequality~\eqref{eq:effmin}}
It finally remains to prove inequality~\eqref{eq:effmin}. One gets
\begin{equation}
\tr\Omega_B^2 \leq \|\Omega_B\| \tr \Omega_B = \|\Omega_B\|
\end{equation}
However,
\begin{eqnarray}
\|\Omega_B\| &=& \sup_{\|\psi_B\| =1} \langle \psi_B | \Omega_B | \psi_B\rangle = \sup_{\|\psi_B\| =1} \langle \psi_B | \tr_S \mathcal{E}_R | \psi_B\rangle
\nonumber\\
&=& \sup_{\|\psi_B\| =1} \sum_{k=1}^{d_S}\langle u_k\otimes \psi_B |  \mathcal{E}_R | u_k \otimes \psi_B\rangle,
\end{eqnarray}
where $\{u_k\}$ is a basis of $\mathcal{H}_S$. From~\eqref{eq:equiprobableR} one gets
\begin{equation}
\langle u_k\otimes \psi_B |  \mathcal{E}_R | u_k \otimes \psi_B\rangle = \frac{1}{d_R} \langle u_k\otimes \psi_B | P_R | u_k \otimes \psi_B\rangle \leq \frac{1}{d_R},
\end{equation}
whence
\begin{equation}
\tr\Omega_B^2  \leq \|\Omega_B\| \leq \frac{d_S}{d_R},
\end{equation}
and~\eqref{eq:effmin} follows.

\begin{figure}[tbp]
\centering
\includegraphics[width=0.4\columnwidth]{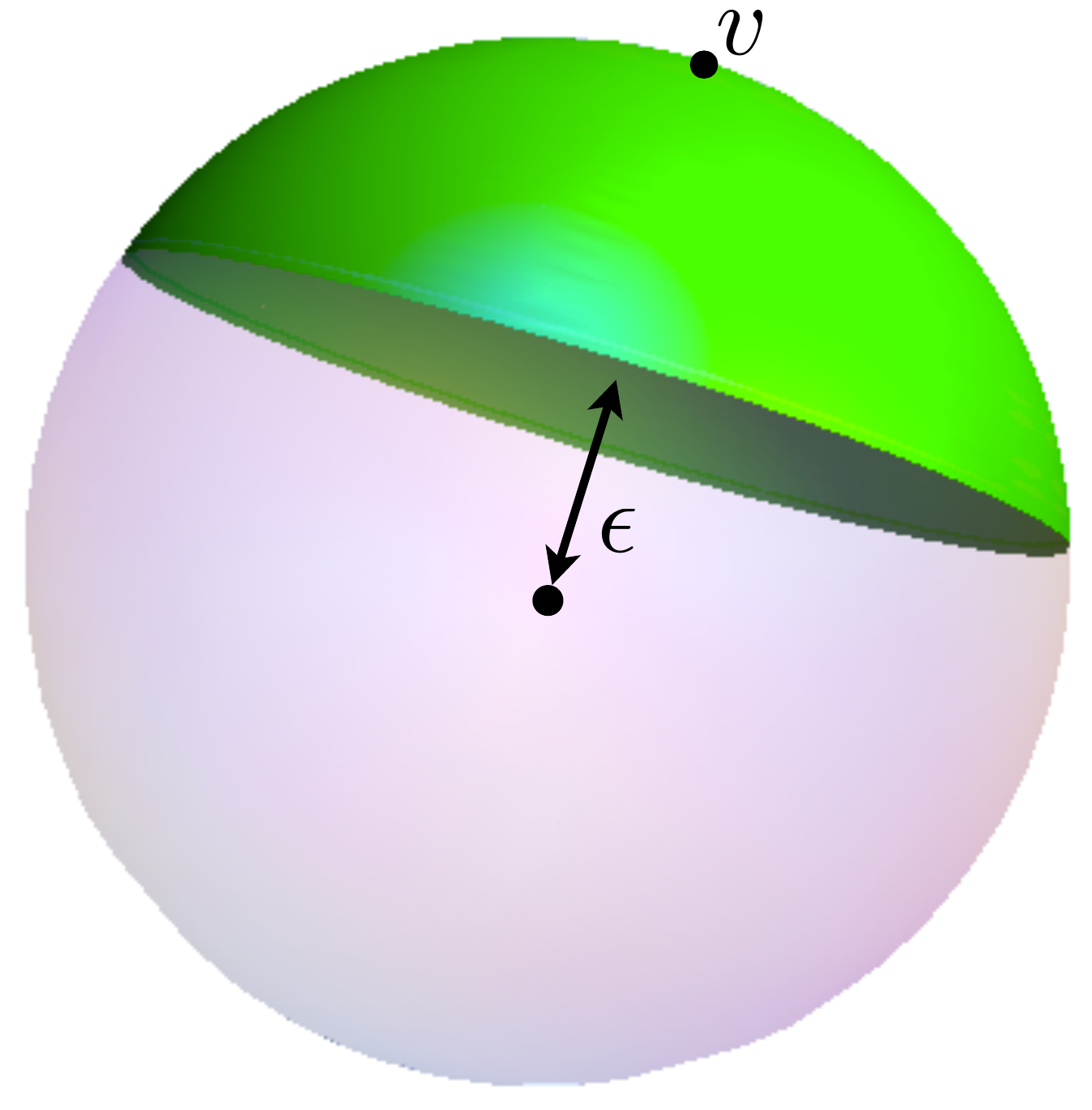}
\caption{Depicted in green an $\varepsilon$-cap around $v$ in $\mathbb{S}^2$. \label{fig:epsiloncap}}
\end{figure}

\section{Lecture 3: L\'evy's Lemma and Convex Geometry}

In this last lecture we are going to present some ideas from convex geometry in high dimensions which have been fruitful in the discussion of canonical typicality. Our goal will be to give a proof of L\'evy's Lemma~\ref{lem:Levy}.
For a deeper  immersion on the subject, see the enjoyable introduction by K.~Ball~\cite{ball}.

\subsection{Concentration of measure in geometry}

The Euclidean unit ball in $\mathbb{R}^n$ will be denoted by $\mathbb{B}^n=\{x\in\mathbb{R}^n\,:\,\sum_{i=1}^n x_i^2\le1\}$, while
its boundary, the unit sphere, by $\mathbb{S}^{n-1}=\partial \mathbb{B}^n=\{x\in\mathbb{R}^n\,:\,\sum_{i=1}^n x_i^2=1\}$.

The measure of $\mathbb{S}^{n-1}$ and of $\mathbb{B}^n$ are related by $|\mathbb{S}^{n-1}|=n\,|\mathbb{B}^n|$, and
one can  explicitly compute
\begin{equation}
|\mathbb{B}^n|=\frac{\pi^{\frac{n}{2}}}{\Gamma\left(\frac{n}{2}+1\right)}.
\end{equation}
By using the Stirling approximation formula on the Euler function $\Gamma$,
\begin{equation}
\Gamma\left(\frac{n}{2}+1\right)\sim\sqrt{2\pi}\e^{-\frac{n}{2}}\left(\frac{n}{2}\right)^\frac{n+1}{2},
\end{equation}
one finds that
\begin{equation}
|\mathbb{B}^n|\sim \left(\frac{2\pi\e}{n}\right)^n,
\end{equation}
as $n\to\infty$. 
This means that the higher the dimension $n$ the smaller is the measure of the Euclidean unit ball. Though it may look highly counterintuitive, it is only one of the strange results one can find in convex geometry. Among those lies  L\'evy's lemma as we are going to discuss.

\begin{figure}[tbp]
\centering
\includegraphics[width=0.6\columnwidth]{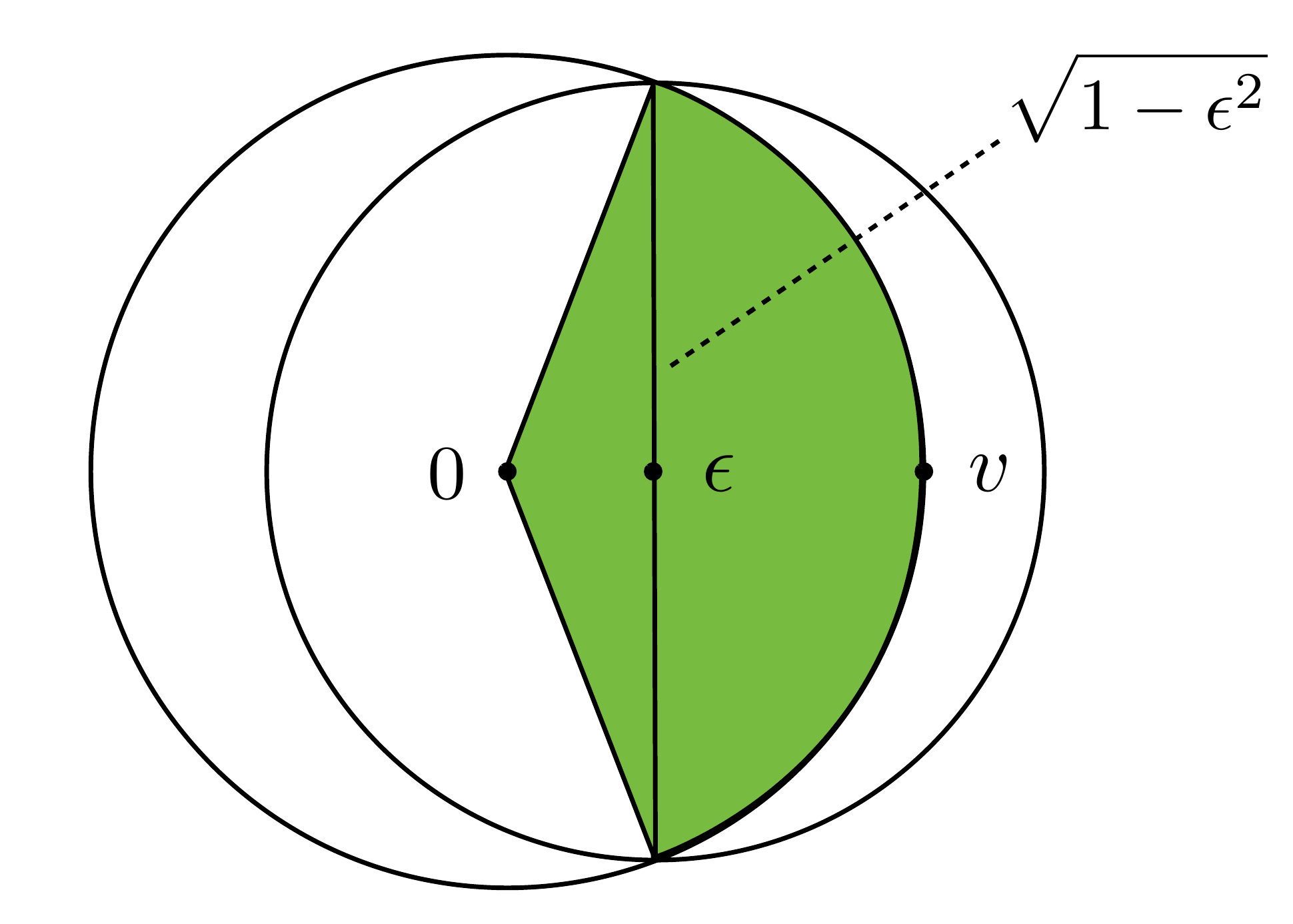}
\caption{Graphical proof of the inequality \eqref{ineq}. $\mathrm{Cone}(\varepsilon,v)$ is depicted in green.\label{fig:ineqfigure}}
\end{figure}

First we need to define what we mean by an $\varepsilon$-cap about a point $v$ on the hypersphere $\mathbb{S}^{n-1}$ (Figure \ref{fig:epsiloncap}). It is the following subset of $\mathbb{S}^{n-1}$:

\begin{equation}
C(\varepsilon,v)=\{\phi\in\mathbb{S}^{n-1}\,:\,\phi\cdot v\ge\varepsilon\},
\end{equation}
where $\cdot$ is the standard scalar product in $\mathbb{R}^n$

Next, define the uniform probability measure on the $n-1$ dimensional sphere as:
\begin{equation}
\sigma_n(A)=\frac{|A|}{|\mathbb{S}^{n-1}|},
\end{equation}
for every measurable set $A\subset \mathbb{S}^{n-1}$.
We are going to prove the following useful lemma~\cite{ball}:
\begin{lemma}
\label{lemma}
\begin{equation}
\sigma_n\big(C(\varepsilon,v)\big)= \frac{|C(\varepsilon,v)|}{|\mathbb{S}^{n-1}|}\le \exp\left(-\frac{n}{2}\varepsilon^2\right)\qquad 0<\varepsilon<1
\end{equation}
\end{lemma}
\begin{proof}
First we recall, by simple geometrical proportionality, that
\begin{equation}
\frac{|C(\varepsilon,v)|}{|\mathbb{S}^{n-1}|}=\frac{|\mathrm{Cone}(\varepsilon,v)|}{|\mathbb{B}^n|}\end{equation}
Next, consider the translated ball as shown in Figure \ref{fig:ineqfigure}. It is evident that as long as $\varepsilon\le 1/\sqrt{2}$ the cone $\mathrm{Cone}(\varepsilon,v)$ is contained into the ball $\mathbb{B}^n(\varepsilon v/\|v\|,\sqrt{1-\varepsilon^2})$  of center $\varepsilon v/\|v\|$ and radius
$\sqrt{1-\varepsilon^2}$, so that
\begin{equation}
\label{ineq}
\frac{|\mathrm{Cone}(\varepsilon,v)|}{|\mathbb{B}^n|}\le\frac{|\mathbb{B}^n(\varepsilon v/\|v\|,\sqrt{1-\varepsilon^2})|}{|\mathbb{B}^n|}.
\end{equation}
The result follows from
\begin{equation}
\frac{|\mathbb{B}^n(\varepsilon v/\|v\|,\sqrt{1-\varepsilon^2})|}{|\mathbb{B}^n|}=\left(1-\varepsilon^2\right)^\frac{n}{2}\le
\exp\left(-\frac{n}{2}\varepsilon^2\right),
\end{equation}
where in the last line we used the elementary inequality: $\ln(1-x)\le-x$.
\end{proof}
\begin{figure}[tbp]
\centering
\includegraphics[width=0.45\columnwidth]{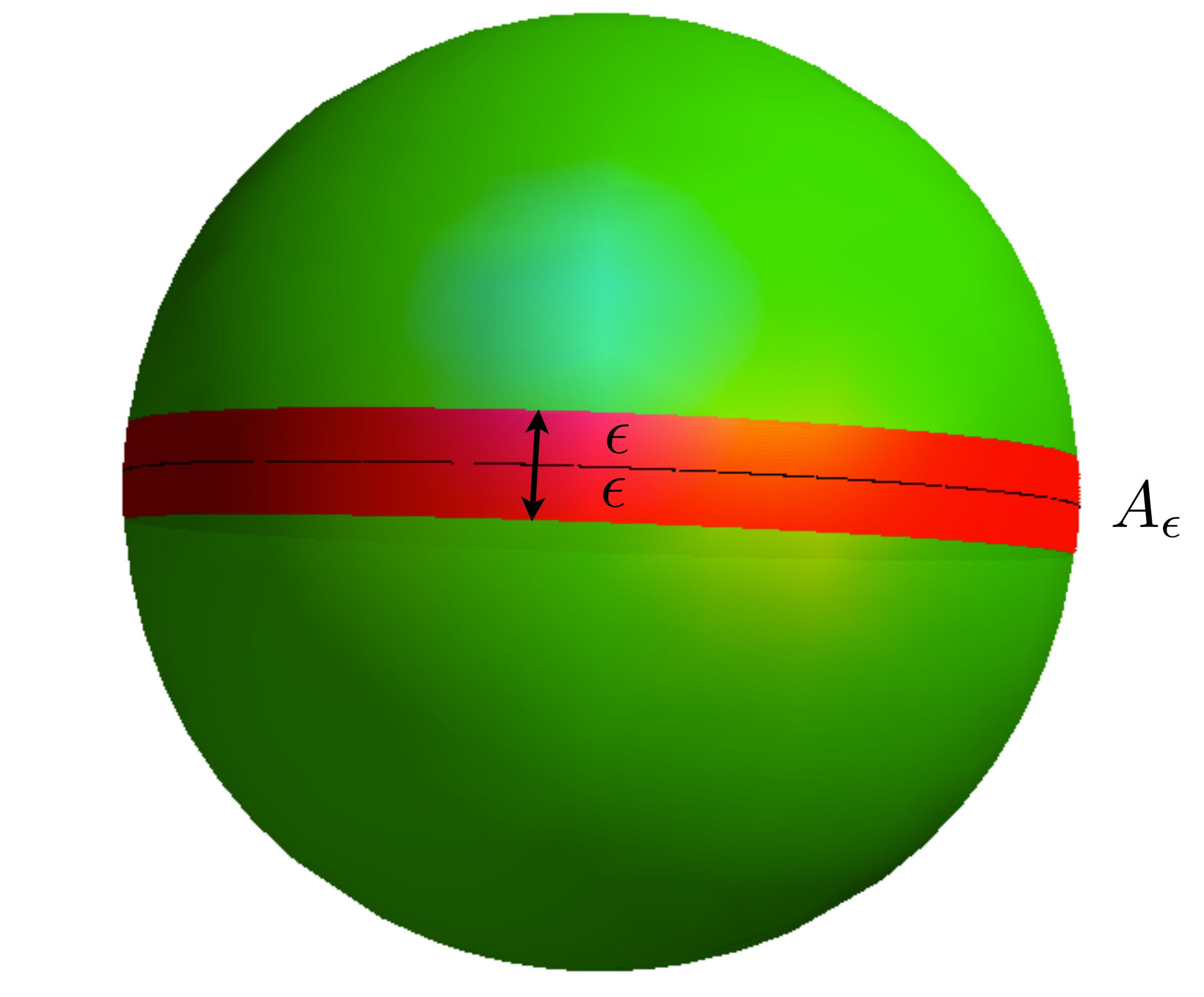}
\caption{A belt, $A_\varepsilon$, around the equator is depicted in red on $\mathbb{S}^2$.\label{belt}}
\end{figure}
Consider a belt around the equator of a sphere, say $A_\varepsilon$ (Figure \ref{belt}), from the previous discussion it follows that
\begin{equation}
\sigma_n(A_\varepsilon)\ge 1-2e^{-\frac{n}{2}\varepsilon^2}.
\end{equation}
From the latter inequality we deduce that the measure is almost concentrated around the equator! This result is quite surprising and goes against our common sense of what happens in low dimensions.

The classical isoperimetric inequality in $\mathbb{R}^n$ states that among all bodies of fixed volume, the Euclidean balls are the ones which have the smallest surface.

Consider a compact set $A$ in $\mathbb{R}^n$. The distance of a point $x$ in $\mathbb{R}^n$ from the set $A$ is 
\begin{equation}
d(x,A)=\inf\{\|x-y\|:y\in A\}, 
\end{equation}where $\|\cdot\|$ is the Euclidean norm. 

Fix $\varepsilon>0$, an $\varepsilon$-neighborhood of the set $A$ is the set 
\begin{equation}
A_\varepsilon=\{x\in \mathbb{R}^n: d(x,A)< \varepsilon \}
\end{equation} 
(Figure \ref{epsilonblow}).
Then, the isoperimetric inequality states that if  the set $A$ and the unit ball $\mathbb{B}^n$ have the same measure, 
\begin{equation}
|A|=|\mathbb{B}^n|,
\end{equation} 
then it follows that \begin{equation}
|A_\varepsilon|\ge|\mathbb{B}^n_\varepsilon|
\end{equation} 
for every $\varepsilon>0$.

This formulation relates the measure and the metric in $\mathbb{R}^n$. In fact if we fatten a set in $\mathbb{R}^n$ into its $\varepsilon$-neighborhood by means of the metric, its measure will increase at least as much as it does for a ball.

So far we have been comparing  neighborhoods of sets by means of two ingredients: the measure and the metrics. The previous discussion can be, then, extended to abstract metric spaces with a  measure.
In particular, we are going to see what happens for the hypersphere $\mathbb{S}^{n-1}$
equipped with, say, the Euclidean distance of $\mathbb{R}^n$ (the geodesic distance will do as well) and the uniform probability measure $\sigma_n$. 

As in $\mathbb{R}^n$, the solutions of the isoperimetric problem on the sphere are the balls in the metric  of $\mathbb{S}^{n-1}$, that is the spherical caps. Hence if $A\subset \mathbb{S}^{n-1}$ is such that  $\sigma_n(A)=\sigma_n(C)$, with $C$ a spherical cap, it follows that
$\sigma_n(A_\varepsilon)\ge\sigma_n(C_\varepsilon)$. 
Though it may seem harmless, the last statement has startling consequences.

\begin{figure}[tbp]
\centering
\includegraphics[width=0.65\columnwidth]{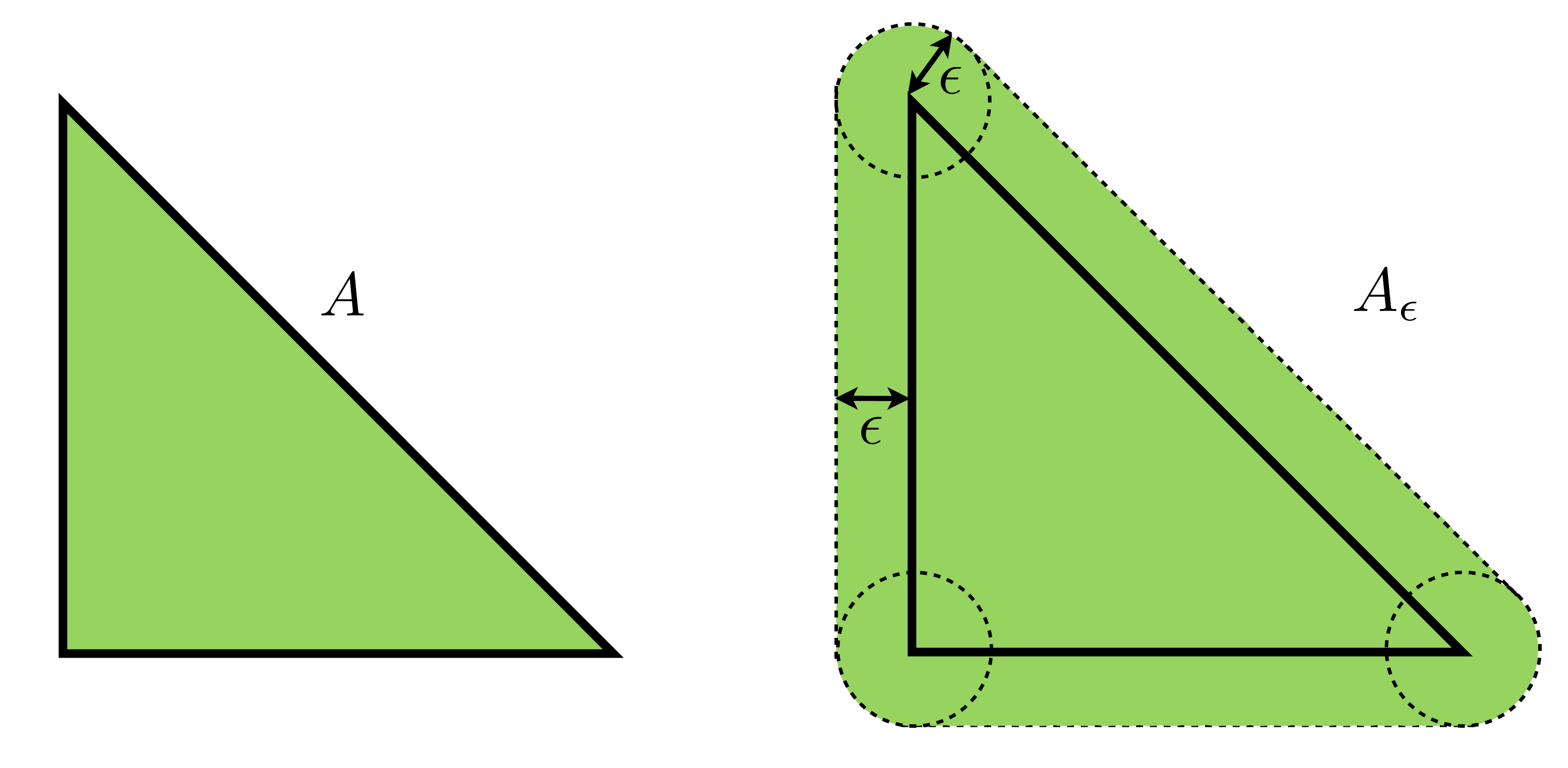}
\caption{An $\varepsilon$-neighborhood of the  triangle $A$.\label{epsilonblow}}
\end{figure}

Fix a set $A$ on $\mathbb{S}^{n-1}$ such that its measure equals the measure of a hemisphere~$H$: 
\begin{equation}
\sigma_n(A)=\sigma_n(H)=\frac{1}{2}.
\end{equation} 
From the isoperimetric inequality it follows that 
\begin{equation}
\sigma_n(A_\varepsilon)\ge\sigma_n(H_\varepsilon), \qquad 0\le\varepsilon\le1. 
\end{equation}
The complement of the fattened hemisphere $H_\varepsilon$ is an $\varepsilon$-cap  $C(\varepsilon,v)$, for some $v\in \mathbb{S}^{n-1}$, that is $C(\varepsilon,v)=\mathbb{S}^{n-1}-H_\varepsilon$. Then, by Lemma \ref{lemma} it follows that: $\sigma_n(C)\le\e^{-\frac{n}{2}\varepsilon^2}$, whence 
\begin{equation}
\sigma_n(A_\varepsilon)\ge \sigma_n(H_\varepsilon)\ge1-\exp\left(-\frac{n}{2}\varepsilon^2\right).
\end{equation} 
This inequality shows that almost the entire sphere lies within a distance $\varepsilon$ of $A$, although there may be some points which are rather far from $A$!
This phenomenon is known as \emph{concentration of measure}: the measure and the metric do not match and the measure $\sigma_n$ of the whole sphere concentrate very close to any set of measure~$1/2$.

Finally, we are going to prove L\'evy's lemma for Lipschitz functions.

\subsection{L\'evy's Lemma}

Suppose $f:\mathbb{S}^{n-1}\to\mathbb{R}$ is a continuous function with Lipschitz constant $\eta=1$, i.e. 
\begin{equation}
|f(\theta)-f(\phi)|\le\|\theta-\phi\|
\end{equation} for every $\theta$ and $\phi$ points on $\mathbb{S}^{n-1}$. 

There is at least one number $m_f\in\mathbb{R}$, called a \emph{median} of $f$, such that both 
\begin{equation}
\sigma_n(A^{-})\ge 1/2 \quad \mathrm{and}\quad  \sigma_n(A^+)\ge 1/2,
\end{equation}
where
\begin{equation}
A^{-}=\{\phi\in\mathbb{S}^{n-1}\,:\,f(\phi)\le m_f\}, \qquad
A^{+}=\{\phi\in\mathbb{S}^{n-1}\,:\,f(\phi)\ge m_f\}.
\end{equation} 
Consider a point $\theta\in A^{-}_{\varepsilon}$, the $\varepsilon$-neighborhood of $A^{-}$, that is  $d(\theta,A^{-})\le\varepsilon$. It follows that 
\begin{equation}
|f(\theta)-m_f|\le\|\theta-\phi_m\|\le \varepsilon, 
\end{equation}where $f(\phi_m)=m_f$. Thus $f(\theta)\le m_f+\varepsilon$ as long as $d(\theta,A^{-})\le\varepsilon$, that is
\begin{equation}
A^{-}_{\varepsilon}\subset \{f(\phi)\le m_f+\varepsilon\} . 
\end{equation}

We claim that only a tiny fraction of the points on the sphere has this property. Indeed,
\begin{equation}
\label{eqlev}
\mathrm{Prob}(f>m_f+\varepsilon)=\sigma_n (\{f>m_f+\varepsilon\})\le 1-\sigma_n(A^{-}_{\varepsilon})\le \exp\left(-\frac{n \varepsilon^2}{2}\right),
\end{equation}
since $\sigma_n(A^{-}_\varepsilon) \geq \sigma_n(H_\varepsilon)$.

Similarly, by considering $A^{+}$, one gets that
\begin{equation}
A^{+}_{\varepsilon}\subset \{f(\phi)\ge m_f-\varepsilon\} ,
\end{equation}
whence
\begin{equation}
\label{eq:lev2}
\sigma_n (\{f<m_f-\varepsilon\})\le \exp\left(-\frac{n \varepsilon^2}{2}\right).
\end{equation} 
By putting together~\eqref{eqlev} and~(\ref{eq:lev2}), L\'evy's lemma follows:
\begin{equation}
\mathrm{Prob}(|f-m_f|>\varepsilon)\le 2 \exp\left(-\frac{n \varepsilon^2}{2}\right).
\label{eq:Levymed1}
\end{equation}
Therefore, the function $f$ is nearly equal to the constant $m_f$ on almost the entire sphere, even if its variation between two antipodal points could be as large as 2.

This result is valid for 1-Lipschitz functions and gives a bound to the deviations of $f$ from its median $m_f$. In order to get the inequality~(\ref{eq:Levy}), one has to consider arbitrary Lipschitz constants $\eta$ and consider the average $\big\langle f\big\rangle$ instead of the median $m_f$. 

As for the first point, notice that if $g$ has Lipschitz constant $\eta$, then $f=\eta^{-1} g$ has Lispchitz constant 1 and $m_f = m_g/\eta$, thus
\begin{equation}
\mathrm{Prob}(|g-m_g|>\varepsilon) = \mathrm{Prob}(|f-m_f|>\varepsilon/\eta )\le 2 \exp\left(-\frac{n \varepsilon^2}{2 \eta^2}\right).
\label{eq:Levymed2}
\end{equation}

As for the second point, notice that if the function is very close to its median for almost all points, its average is also very close to the median  $\big\langle f\big\rangle \approx m_f$, except for an exceptional set of exponentially small measure. Thus one obtains an inequality of the same form as~\eqref{eq:Levymed2}, with a different constant in the exponent, namely~\eqref{eq:Levy}.

\section*{Acknowledgments}

We would like to thank the  organizers, A.P.~Balachandran, Beppe Marmo, and Sachin Vaidya for their kindness in inviting us and for the effort they exerted on the organization of the workshop. 
This work was partially supported by INFN through the project ``QUANTUM'' and by the Italian National Group of Mathematical Physics (GNFM-INdAM).


\begin{thebibliography}{99}
\bibitem{schr1}E. Schr\"odinger, \textit{Statistical Thermodynamics}, 1989
\bibitem{gemmer} J. Gemmer, M. Michel and G. Mahler, \textit{Quantum Thermodynamics}, (Springer-Verlag, Berlin, 2009).
\bibitem{goldstein} S. Goldstein, J.L. Lebowitz, R. Tumulka and N. Zangh\`i, \textit{Phys. Rev. Lett.}, 96, 050403 (2006) 
\bibitem{popescu} S. Popescu, A.J. Short and A. Winter, arXiv:quant-ph/0511225, 2005
\bibitem{popescun} S. Popescu, A.J. Short and A. Winter, \textit{Nature Physics}, 2, 754-758 (2006).
\bibitem{uffink} J. Uffink, \textit{Handbook for Philsophy of Physics} 924-1074 (Elsevier, Amsterdam , 2007).
\bibitem{Boltz}  L. Boltzmann, Studien \"uber das Gleichgewicht der lebendigen Kraft zwischen bewegten materiellen Punkten, \textit{Wiener Berichte}, 58, 517-560
\bibitem{gibbs} J.W. Gibbs, \textit{Elementary Principles in Statistical Mechanics}, (Scribner, New York, 1902)
\bibitem{oliveira} C.R. de Oliveira, T. Werlang, \textit{Rev. Bras. Ensino. Fis.}, 29, 189-201, 2007.
\bibitem{prentis} J.J. Prentis, \textit{Am. J. Phys}, 68, 1073, 2000.
\bibitem{gallavotti} G. Gallavotti, \textit{Statistical Mechanics. A short treatise} (Springer-Verlag, Berlin, 1999).
\bibitem{Huang} K. Huang, \textit{Statistical Mechanics}, (Wiley \& Sons, 1987)
\bibitem{landau} L.D. Landau, E.M. Lifshitz, \textit{Statistical Physics}, (Pergamon, London, 1958)
\bibitem{schr}E. Schr\"odinger, \textit{Proc. Cambridge Philos. Soc.}, 31, 555-563, 1935.
\bibitem{nielsen} M.A. Nielsen, I.L. Chuang, \textit{Quantum Computation and Quantum Information} (Cambridge University Press, Cambridge, 2000).
\bibitem{vn} J. von Neumann, \textit{Mathematical Foundation of Quantum Mechanics} (Princeton University Press, Princeton, 1955)
\bibitem{shannon} C.E. Shannon,  \textit{Bell Syst. Tech. J.}, 27, 1948
\bibitem{vn1} J. von Neumann, \textit{European Phys. J. H} 35, 201-237 (2010). German original in \textit{Zeitschrift fuer Physik} 57, 30-70 (1929)
\bibitem{ball} K. Ball, \textit{An Elementary Introduction to Modern Convex Geometry}, in 
\textit{Flavors of Geometry} MSRI Publications, Vol. \textbf{31}, 1-58 (1997).
\end{thebibliography}
\end{document}